\documentclass[12pt, a4paper]{amsart}

\usepackage[matrix,arrow,curve,frame]{xy}
\usepackage{amsmath,amsthm,amssymb,enumerate}
\usepackage{latexsym}
\usepackage{amscd}
\usepackage[backref]{hyperref}
\usepackage{euscript}
\usepackage{bbm}
\usepackage{tikz}
\usepackage{tikz-cd}
\usepackage[multiple]{footmisc}

\usepackage{fullpage}

\newtheorem{theorem}{Theorem}[section]
\newtheorem{lemma}[theorem]{Lemma}

\newtheorem{proposition}[theorem]{Proposition}

\theoremstyle{definition}
\newtheorem{definition}[theorem]{Definition}
\newtheorem{example}[theorem]{Example}

\theoremstyle{remark}
\newtheorem{remark}[theorem]{Remark}

\numberwithin{equation}{section}
\newcommand{\beq}{\begin{equation}}
\newcommand{\eeq}{\end{equation}}

\newcommand{\ZZ}{\mathbb{Z}}

\def\ra{\rightarrow}

\newcommand {\be}{\begin{equation}}
\newcommand {\ee}{\end{equation}}
\newcommand{\h}{\begin{eqnarray*}}
\newcommand{\e}{\end{eqnarray*}}

\makeatletter
\newcommand{\leqnomode}{\tagsleft@true}
\newcommand{\reqnomode}{\tagsleft@false}
\makeatother

\begin{document}

\title
{Exotic Courant algebroids and T-duality\\}

 \author{Jaklyn Crilly}
\address{Mathematical Sciences Institute,
Australian National University, Canberra 0200, Australia}
\email{jaklyn.crilly@anu.edu.au}

 \author{Varghese Mathai}
\address{School of Mathematical Sciences,
University of Adelaide, Adelaide 5005, Australia}
\email{mathai.varghese@adelaide.edu.au}

\date{}

\maketitle

\begin{abstract}
In this paper, we extend the T-duality isomorphism in
\cite{cavalcanti} from invariant Courant algebroids, to exotic Courant algebroids such that the {\em momentum} and {\em winding numbers} are exchanged,
filling in a gap in the literature. 
\end{abstract}

\section*{Introduction}
\reqnomode

In \cite{cavalcanti}, Gualtieri and Cavalcanti showed that T-duality for principal circle bundles in a background $H$-flux gives an isomorphism between invariant Courant algebroids. The goal in this paper is to extend their isomorphism to one mapping from the full standard Courant algebroid defined on
a principal circle bundle. The striking answer obtained is that the T-dual is an {\em exotic} Courant algebroid (defined in the paper) on the T-dual 
principal circle bundle, and we conclude that the {\em momentum} and {\em winding numbers} are exchanged. This paper was inspired by the works of \cite{HM14,HM18}, who extended the isomorphism of \cite{BEM,BEM2} to the full space of differential forms on a principal circle bundle, with the T-dual then being the space of {\em exotic} differential forms on the T-dual principal circle bundle.

In more detail, recall that topological T-duality as in \cite{BEM,BEM2} asserts that if $\pi : Z \to M $ denotes a principal circle bundle whose first Chern class is given by $[F] \in H^2(M, \mathbb{Z})$, and $[H] \in H^3(Z, \ZZ)$ denotes a $H$-flux on $Z$, then there exists a T-dual bundle $\hat{\pi} : \hat{Z} \to M$ whose first Chern class is denoted $[\hat{F}] \in H^2(M, \mathbb{Z})$ along with a T-dual $H$-flux on the T-dual bundle, $[\hat{H}] \in H^3(\hat{Z}, \ZZ)$, satisfying
$$
[\hat{F}] = \pi_*([H]),\qquad
[F] = \hat{\pi}_*([\hat{H}]).
$$
The result of \cite{cavalcanti} asserts that there is an isomorphism of Courant algebroids over $M$,
\begin{align}
\label{CG}
\mu_0: \Gamma(Z, TZ \oplus T^*Z)^{S^1} \longrightarrow \Gamma(\hat Z, T\hat Z \oplus T^*\hat Z)^{\hat S^1}.
\end{align}
We extend this isomorphism which is associated to the invariant data, to an isomorphism taking into account the non-invariant data too, that is, a T-dual isomorphism mapping from the standard Courant algebroid $TZ \oplus T^*Z$ defined over $Z$. We achieve this by first showing that the standard Courant algebroid over the principal circle bundle $Z$ is equivalent to an {\it exotic} Courant algebroid defined over the base space $M$, and then proving that there is a graded T-duality isomorphism of exotic Courant algebroids defined over $M$, which when restricted to the invariant sections of $TZ \oplus T^*Z$ over $Z$, is exactly the T-duality isomorphism (\ref{CG}). It is thus shown in Theorem \ref{T-duality-exotic} that this isomorphism, given by
\begin{align}
\label{muuu}
\mu: \Gamma(Z, TZ \oplus T^*Z) \longrightarrow \bigoplus_{n\in\ZZ}\Gamma \big( \hat{Z} , (T\hat Z \oplus T^*\hat Z) \otimes {\hat\pi}^*(L^{\otimes n})\big)^{\hat S^1}
\end{align}
for the line bundle $L$ associated to $Z$, defines an isomorphism between each of the spaces underlying the exotic Courant algebroid structures (to be defined in section 2.2). 

%
%
\noindent\textbf{Acknowledgements}.   The authors thank the participants (especially Maxim Zabzine) of the conference entitled, {\em Index theory and related fields}, Chern Institute of Mathematics, Tianjin
June 17-21, 2019, for feedback on first author's talk over there. Mathai thanks Fei Han for earlier collaboration that inspired this work and both authors thank David Baraglia for helpful initial discussions.
Jaklyn Crilly acknowledges her Laureate PhD scholarship support through FL170100020 and Varghese Mathai thanks the Australian Research Council  for support via the Australian Laureate Fellowship FL170100020.
\section{Preliminaries}\label{sect:courant}
In this section we recall the notions of a Courant algebroid, along with some standard, but important, examples. For some excellent articles and reviews on Courant algebras in the literature see \cite{Baraglia,Baraglia2,Bouwknegt07,Courant,gualtieri, Hekmati,Hitchin, Hitchin2,Li-Bland,KU,Zabzine}. 
We then review topological T-duality as in \cite{BEM} and see how the T-duality transformations can be applied to certain Courant algebroid as in \cite{cavalcanti}, before reviewing the extended T-duality isomorphism for exotic differential forms \cite{HM18}.

\subsection{Courant algebroids}
Given a manifold $M$, let the data $(E, \langle  \,, \rangle, [ \,, ], \rho)$ consist of a (possibly infinite dimensional) vector bundle $E \ra M$, a nondegenerate, symmetric bilinear form $\langle \cdot ,\cdot \rangle: \Gamma(E) \times  \Gamma(E) \ra C^{\infty}(M)$, a bilinear ({\it Dorfman}) bracket $[\cdot ,\cdot]: \Gamma(E) \times  \Gamma(E) \ra \Gamma(E)$, and a smooth bundle map $\rho: E \ra TM$ which we term the {\it anchor map}.

\begin{definition}
A Courant algebroid over $M$ consists of the data $(E, \langle  \, ,\,  \rangle, [\, ,\,  ], \rho)$ defined over $M$, which is compatible with the following conditions:
\begin{enumerate} 
\item 
$[a, [b, c]] = [[a, b ], c] + [b, [a, c]]$,
\item 
$\rho([a, b]) = [\rho(a) , \rho(b)],$
\item
$[a, h b] = \rho(a)(h) b +h [a, b]$,
\item
$[a, b] + [b, a] =d \langle a , b \rangle$,
\item
$\rho(a) \langle b , c \rangle = \langle [a, b] , c \rangle + \langle b , [a, c] \rangle$.
\end{enumerate}
where $a, b, c \in \Gamma(E)$, $h \in C^{\infty}(M)$, and $d: C^{\infty}(M) \ra \Gamma(E)$ is the induced differential operator defined by the relation: $\langle dh, a\rangle = \rho (a) h.$
\end{definition}

\begin{remark}
In keeping with the literature, we have included conditions (2) and (3). It can be shown however that they are a consequence of the non-degeneracy of the bracket and condition (5), and are thus redundant.
\end{remark}

\begin{example}
\label{stand}
The {\it standard H-twisted Courant algebroid} consists of the vector bundle $TM \oplus T^*M \to M$ given by the direct sum of the tangent and cotangent bundles over $M$, along with the following defining structures\footnote{The Dorfman bracket can alternatively be defined as the derived bracket of the $H$-twisted differential $d+H$ acting on $\Omega^*(M)$ (that is, $ [ \cdot , \cdot ]_H := [[d+H, \, \cdot \, \,], \, \cdot \, \, ]$), where the sections of $TZ\oplus T^*Z$ act on $\Omega^*(M)$ via the Clifford action given by:
\label{cliffordmod} $(X+\xi)\cdot \omega = \iota_X (\omega) + \xi \wedge \omega.$}:
\begin{align*}
\bullet \, \, \, & [X + \alpha, Y + \beta]_H := [X,Y] + \mathcal{L}_X\beta - \iota_Y d \xi + \iota_X \iota_Y H,\\
\bullet \, \, \, & \langle X + \alpha, Y + \beta \rangle := \frac{1}{2}(\alpha(Y) + \beta(X)),\\
\bullet \, \, \, & \rho(X+\alpha) := X,
\end{align*}
where $X, Y \in \Gamma(TM)$, $\alpha, \beta \in \Omega^1(M)$, $H \in H^3(M)$, $\mathcal{L}_X$ denotes the Lie derivative along $X$, and $\iota_Y H$ denotes the interior product of $H$ with $X$.
\end{example}

The next example will be that of an {\it invariant Courant algebroid}. This is as a non-exact Courant algebroid defined over a principal circle bundle and will be of considerable interest when exploring T-duality. Henceforth, we will let $\pi : Z \to M$ denote a principal circle bundle, $v$ denote an invariant period-1 generator of the circle action on $Z$, fix an invariant connection form $A \in \Omega^1(Z)$ normalized such that $\iota_v A := A(v)=1$, and let
$$\Gamma(TZ \oplus T^*Z)^{S^1} := \{x+\alpha \in \Gamma(Z,TZ \oplus T^*Z) | \mathcal{L}_v(X + \alpha)=0 \},$$
where $\mathcal{L}_v$ denotes the Lie derivative along $v$, $x \in \Gamma(Z,TZ)$, and $\alpha \in \Gamma(Z,T^*Z)$, denote the invariant sections of the bundle $TZ \oplus T^*Z$ over $Z$.

Now for any invariant section $x + \alpha \in \Gamma(TZ \oplus T^*Z)^{S^1}$, there exists a unique $X \in \Gamma(M, TM)$, $\xi \in \Omega^1(M)$, and $f, g \in C^{\infty}(M)$ such that 
\begin{align}
x + \alpha = h_A(X)+fv +\pi^*(\xi) + gA  \in \Gamma(TZ \oplus T^*Z)^{S^1},
\end{align}
where $h_A$ denotes the horizontal lift with respect to the connection $A$, and $\pi^*$ the pullback with respect to $\pi : Z \to M$.

Thus, there is a bijective correspondence between such invariant sections over $Z$ and arbitrary sections of the bundle $E:=TM \oplus \mathbbm{1}_{\mathbb{R}} \oplus T^*M \oplus \mathbbm{1}_{\mathbb{R}}$ over $M$, given by:\footnote{In keeping with the literature, instead of writing $h_A(X)$ we will simply write $X$. Thus $X$ will mean both the vector field on $TM$ as well as its horizontal lift $h_A(X)$. Likewise for the basic forms.}\footnote{Given the bijection $\phi$, we will often write a section over $Z$ as $X+fv +\xi + gA$, when what is really meant is $(X, f, \xi, g)$. This emphasizes that the data of interest is related to the bundle $Z$.}
\begin{align}
\label{invsec}
\phi \, : \, \, \Gamma(M, E)\, \, &\mapsto \, \Gamma(Z, TZ \oplus T^*Z)^{S^1} \\
(X, f, \xi, g) &\mapsto X+fv +\xi + gA. \nonumber
\end{align}

\begin{example}[Invariant Courant algebroid]
\label{ICA1}
Consider the $H$-twisted standard Courant algebroid over the principal circle bundle $Z$ where we will take $H$ to be an invariant 3-form on $M$, given by $H = H_{(3)} + A \wedge H_{(2)}$. We want to restrict all the structures so that they are only defined on the invariant sections, i.e. 
\begin{align*}
[ \, \, , \,]_H &: \Gamma(TZ \oplus T^*Z)^{S^1} \times \Gamma(TZ \oplus T^*Z)^{S^1} \to \Gamma(TZ \oplus T^*Z)^{S^1}\\
\langle \, \, , \, \rangle \, \, \, &:  \Gamma(TZ \oplus T^*Z)^{S^1} \times \Gamma(TZ \oplus T^*Z)^{S^1} \to C^{\infty}(M),
\end{align*}
where we have observed that the invariant sections are closed under the Dorfman bracket.

Now using the bijection $\phi$ we can transfer the standard Courant algebroid structures over $Z$ (restricted to the invariant sections) to a Courant algebroid structure over $M$ with vector bundle $E=TM \oplus \mathbbm{1}_{\mathbb{R}} \oplus T^*M \oplus \mathbbm{1}_{\mathbb{R}}$ defined over $M$. Doing so, we get that the anchor map is given by $\bar{\rho}:= \pi_* \circ \rho : E \overset{\rho}{\to} TZ \overset{\pi_*}{\to} TM$, whilst the remaining structures defined on the space of sections $\Gamma(M,E):=\Gamma(M,TM \oplus \mathbbm{1}_{\mathbb{R}} \oplus T^*M \oplus \mathbbm{1}_{\mathbb{R}})$ are:
\begin{align*}
\bullet & \, \, \, \langle (X,f, \alpha, g), (Y ,\tilde{f} , \beta , \tilde{g} ) \rangle := \beta(X) + \alpha(Y) + g\tilde{f} + f\tilde{g}\\
\bullet & \, \, \, [(X,f , \alpha , g ), (Y ,\tilde{f} , \beta , \tilde{g} )]_H := \phi^{-1} \big( [\phi(X,f , \alpha , g), \phi(Y ,\tilde{f} , \beta , \tilde{g} )]_H \big)
\end{align*}

Then $E$ defines a non-exact Courant algebroid over $M$, denoted by $\big((TZ \oplus T^*Z)^{S^1}, [\cdot , \cdot ]_H \big)$, encoding the invariant data of the $H$-twisted, standard Courant algebroid over $Z$.
\end{example}

\subsection{Topological T-duality}
We are going to be interested in topological T-duality arising for the case of principal circle bundles with a H-flux. For such a case, one begins with a principal circle bundle $\pi: Z \to M$ whose first Chern class is given by $[F] \in H^2(M, \mathbb{Z})$, along with a H-flux which is given by some $[H] \in H^3(Z, \mathbb{Z})$. The aim is to then determine the corresponding data arising after an application of the T-duality transformation. 

One method for determining the T-dual data is by focusing on the Gysin sequence associated to the bundle $\pi : Z \to M$, given by:
\begin{center}
\begin{tikzcd}
\cdots \arrow[r] & H^{3}(M, \mathbb{Z}) \arrow[r, "\pi^{*}"] & H^3(Z, \mathbb{Z}) \arrow[r, "\pi_{*}"] & H^{2}(M, \mathbb{Z}) \arrow[r, "\lbrack F \rbrack \wedge"] & H^{4}(M, \mathbb{Z})  \arrow[r] & \cdots
\end{tikzcd}
\end{center}
Then letting $[H] \in H^3(Z, \mathbb{Z})$, define $[\hat{F}] = \pi_*([H]) \in H^2(M, \mathbb{Z})$, choose and fix a principal circle bundle $\hat{\pi} : \hat{Z} \to M$ whose first Chern class is given by $[\hat{F}]$. Having made such a choice, consider the Gysin sequence associated to the bundle $ \hat{Z}$ over $ M$,
\begin{center}
\begin{tikzcd}
\cdots \arrow[r] & H^{3}(M, \mathbb{Z}) \arrow[r, "\hat{\pi}^{*}"] & H^3(\hat{Z}, \mathbb{Z}) \arrow[r, "\hat{\pi}_{*}"] & H^{2}(M, \mathbb{Z}) \arrow[r, "\lbrack \hat{F}\rbrack \wedge"] & H^{4}(M, \mathbb{Z})  \arrow[r] & \cdots
\end{tikzcd}
\end{center}
Now using exactness of the above exact sequence, and the fact that $[F] \wedge [\hat{F}]= [F] \wedge \pi_*([H])=0$, there exists a $[\hat{H}] \in H^3(\hat{Z})$ such that $[F] = \hat{\pi}_*([\hat{H}])$. As shown in \cite{BEM}, the pairs $([F], [H])$ and $([\hat{F}], [\hat{H}])$ are T-dual (maps between them defining the T-duality transformation), and are unique up to bundle automorphism. This provides a global, geometric version of the Buscher rules \cite{Buscher}.

A second useful and relevant T-duality transformation derived from the same paper \cite{BEM} is the following:
%
%

\begin{theorem}[\cite{BEM}]\label{TD2}
Let $A$, $\hat{A}$ denote connection forms on $Z$ and $\hat{Z}$ respectively, choose invariant representatives $H \in [H]$ and $\hat{H} \in [\hat{H}]$, and let $\big(\Omega^{*}(Z)^{S^1}, d + H\big)$ be the $H$-twisted, $\mathbb{Z}_2$-graded differential complex consisting of invariant differential forms over $Z$.\footnote{Here we have $\Omega^{*}(Z)^{S^1} := \{\omega \in \Omega^{*}(Z) |\mathcal{L}_v(\omega) = 0\},$ where $v$ is the invariant period-1 generator of the circle action on $Z$, and $\mathcal{L}_v$ denotes the Lie derivative along $v$.}

Then the following map:
\begin{align*}
T: (\Omega^{*}(Z)^{S^1}, d+H) &\to (\Omega^{*+1}(\hat{Z})^{\hat{S}^1}, -(d+{\hat{H})})\\
\omega \quad &\mapsto \int_{S^1}\omega \wedge e^{ \hat{A} \wedge A},
\end{align*}
is a chain map isomorphism between twisted, $\mathbb{Z}_2$-graded complexes.

Furthermore, this induces an isomorphism on twisted cohomologies:
\begin{align*}
T: H^{*}_{d+H}(Z) &\to H^{*+1}_{d+\hat{H}}(\hat{Z}).
\end{align*}
Such maps define the T-duality transformation between the relevant sets of data.
\end{theorem}

\subsection{Exotic differential forms}
In \cite{HM18}, the T-duality isomorphism mapping between the invariant differential forms defined on the given T-dual bundles as given above in Theorem \ref{TD2} was extended to an isomorphism mapping from the full space of complex-valued differential forms. 
In order to define this T-duality mapping, let $L, \hat L$ denote the complex line bundles associated to the circle bundles $Z, \hat Z$ with the standard representation of the circle on the complex plane respectively. The {\em exotic differential forms} (on $Z$ and $\hat{Z}$ respectively) are then given by
$$\mathcal{A}^{\bar k}(Z)^{S^1}=\bigoplus_{n\in \ZZ}\mathcal{A}^{\bar k}_n(Z)^{S^1}:=\bigoplus_{n\in \ZZ}\Omega^{\bar k}(Z, \pi^*(\hat{L}^{\otimes n}))^{S^1},$$
$$\mathcal{A}^{\bar k}(\hat{Z})^{\hat S^1}=\bigoplus_{n\in \ZZ}\mathcal{A}^{\bar k}_n(\hat{Z})^{\hat S^1}:=\bigoplus_{n\in \ZZ}\Omega^{\bar k}(\hat{Z}, \hat{\pi}^*(L^{\otimes n}))^{\hat S^1}$$
for $\bar k= k \mod 2$, and where we have taken the direct sum above to be the Fr\'{e}chet space completion of the standard direct sum. This definition of the direct sum (as a completion) will be the implicit definition from here on out when using direct sums in the context of the exotic structures. This is notationally consistent with \cite{HM18}.

{ Now define the subspace of $n$th-weight differential forms on $Z$ to be given by,
\begin{align}
 \Omega^*_{n}(Z):=\{\omega \in \Omega^*(Z)| \mathcal{L}_{v}\omega=n\omega\},
\end{align}
and observe that
$$\Omega^{\bar k}_0(Z)=\Omega^{\bar k}(Z)^{S^1}, \ \  \mathcal{A}^{\bar k+1}_0(\hat{Z})^{\hat S^1}=\Omega^{\overline{k+1}}(\hat Z)^{\hat S^1}.$$
Then under a particular choice of Riemannian metrics and flux forms, the results of \cite{HM18} show that the Fourier-Mukai transform $T$ in Theorem \ref{TD2} can be extended to a sequence of  isometries
\begin{align}
\tau_n \colon \Omega^{\bar k}_{-n}(Z) \to \mathcal{A}^{\bar k+1}_n(\hat{Z})^{\hat S^1},
\end{align}
defined by the {\em exotic Hori formula} from $Z$ to $\hat Z$ given in \cite{HM18} for $\bar k= k \mod 2$, where the twisted de Rham differential $d+H$ maps to the differential $-(\hat{\pi}^*\nabla^{L^{\otimes n}}-\iota_{n\hat{v}}+\hat{H})$, and we observe that $\tau_0=T$.  { One similarly has 
a sequence of  isometries,
\begin{equation}
\sigma_n \colon\mathcal{A}^{\bar k}_n(Z)^{ S^1} \to \Omega^{\bar k +1}_{-n}(\hat{Z}),
\end{equation}
defined by the {\em inverse exotic Hori formula} form $Z$ to $\hat Z$ given in \cite{HM18} for $\bar k= k \mod 2$, where the differential $\pi^*\nabla^{\hat L^{\otimes n}}-\iota_{n{v}}+{H}$
maps to  the twisted de Rham differential $-(d+{\hat H})$, and $\sigma_0=T$. Similarly, one can define the sequences of isometries $\hat\tau_n, \hat\sigma_n$ on $\hat Z$. Although the extension of the Fourier-Mukai transform 
to all differential forms on $Z$ is slightly asymmetric, one has the following crucial identities, verified in \cite{HM18}:
\begin{align}
{\rm -Id} =\hat\sigma_n \circ \tau_n \colon \Omega^{\bar k}_{-n}(Z)  \longrightarrow \Omega^{\bar k}_{-n}(Z),\\
{\rm -Id} =\hat\tau_n \circ \sigma_n  \colon \mathcal{A}^{\bar k}_n(Z)^{ S^1} \longrightarrow \mathcal{A}^{\bar k}_n(Z)^{ S^1}.  
\end{align}
This is interpreted as saying that T-duality, when applied twice, returns the object to minus itself, which arises due to the convention of integration along the fiber. This was a result previously verified in \cite{BEM,BEM2} for the special case of when $n=0$.

This shows that for each of either $Z$ or $\hat Z$, there are two theories (at degree 0 the two theories coincide), and there are also graded isomorphisms between the two theories of both sides. }

Moreover, when $n\neq 0$ 
the complex $(\mathcal{A}^{\bar k+1}_n(\hat{Z})^{\hat S^1}, \hat{\pi}^*\nabla^{L^{\otimes n}}-\iota_{n\hat{v}}+\hat{H})$ has vanishing cohomology. Therefore, when $n\neq 0$ the complex $(\Omega^{\bar k}_{-n}(Z), d+H)$ also has 
vanishing cohomology. In \cite{HM18},  an explicit homotopy is constructed to show this.


\subsection{T-duality of Courant algebroids}
T-duality has a natural extension to a duality between certain Courant algebroids defined on principal circle bundles. To see this, consider again a principal circle bundle $\pi : Z \to M$ with invariant $H$-flux representative $H \in \Omega^3( Z)^{ S^1}$, and let $\hat \pi: \hat Z \to M$ denote the T-dual principal circle bundle over $M$ with T-dual $H$-flux representative $\hat H \in \Omega^3(\hat Z)^{\hat S^1}$.

\begin{theorem}[\cite{cavalcanti}]\label{t-duality-courant}
\label{cg}
Consider the invariant Courant algebroids defined over $M$ given by $((TZ\oplus T^*Z)^{S^1}, [\cdot, \cdot]_H)$ and $((T\hat Z\oplus T^*\hat Z)^{\hat S^1}, [\cdot, \cdot]_{\hat H})$ (as introduced in Example \ref{ICA1}). Then there exists a Courant algebroid isomorphism between these objects, given by
\begin{align*}
\chi \, \, : \, \, \, \, (TZ\oplus T^*Z)^{S^1} \, \, \, &\to \, \, \, (T\hat Z\oplus T^*\hat Z)^{\hat S^1}\\
(X + f v + \xi + g A) &\to (X + g\hat v + \xi + f \hat A)
\end{align*}
where $X \in \Gamma(M, TM)$, $\xi \in \Omega^1(M)$, and $f, g \in C^{\infty}(M)$.

Furthermore, this map defines an isomorphism between Clifford algebras, and so
$$T(a \cdot \omega)=\chi(a) \cdot T(\omega),$$
where $a \in \Gamma(TZ \oplus T^*Z)^{S^1}$, $\omega \in \Omega^*(Z)^{S^1}$, and the map $T$ is as defined in Theorem \ref{TD2}.
\end{theorem}

Our goal is to generalize this T-duality isomorphism  from the invariant, $H$-twisted Courant algebroid $(TZ\oplus T^*Z)^{S^1}$ over $M$, to an isomorphism from the standard, $H$-twisted Courant algebroid $(TZ\oplus T^*Z)$ over $Z$. We mention that another generalization of Theorem 1.10 to the invariant chiral
de Rham complex was established in \cite{LM15}.

\section{Exotic  Courant algebroids}
In this section, we introduce the novel concept of exotic Courant algebroids (see Definition \ref{exotic courant}), which plays a central role in our formulation of T-duality in the last section. 

\subsection{Exotic Courant algebroids}
The objects we are going to be interested in are infinite dimensional bundles which are constructed from a (complex) Courant algebroid over $M$, denoted $(E, \langle  \,, \rangle, [ \,, ], \rho)$,  and a (complex) line bundle $L\to M$ possessing a connection $\nabla : \Gamma(M,L) \to \Gamma(M,L \otimes T^*M)$. In particular, we will be interested in the bundle
$$\pi :  \underset{n \in \mathbb{Z}}{\bigoplus} \big(E \otimes L^{\otimes n}\big) \to M,$$
on which we will define an {\it exotic Courant algebroid}, i.e., a generalization of the Courant algebroid structure to this infinite dimensional bundle. To do this we need to define an {\it exotic} bracket, product and anchor map, but before this can be done, an additional structure on the infinite dimensional bundle is required in order to make sense of the exotic anchor map. 

For a Courant algebroid, the anchor map is a bundle map $\rho: E \ra TM$, where the tangent bundle $TM$ over $M$, is endowed with a Lie algebroid structure (which by definition defines an action on $C^{\infty}(M)$). Similarly, for the case of the exotic Courant algebroid, we will define our anchor map to be a bundle map $\rho := \underset{n \in \mathbb{Z}}{\oplus} \rho_n$, where
$$\rho_n : E \otimes L^{\otimes n} \to TM \otimes L^{\otimes n},$$
such that the bundle $\underset{n \in \mathbb{Z}}{\bigoplus}\big(TM \otimes L^{\otimes n}\big)\to M$ is endowed with a structure consisting of an action on $\Gamma(L^{\otimes m})$ for all $m \in \mathbb{Z}$, and a bracket. Such a structure is defined as follows:

\begin{itemize}

\item {\underline{Action} :} Let $a:=X \otimes s \in \Gamma\left(M,TM \otimes L^{\otimes n}\right)$ and $h \in \Gamma(M, L^{\otimes p})$ denote sections over $M$. Then we define the action of sections of the bundle $TM \otimes L^{\otimes m} \to M$ on sections of the bundle $L^{\otimes m} \to M$ as follows:
$$a (h) := \iota_X(\nabla^{\otimes p} h) \otimes s.$$\\

\item {\underline{Bracket} :} Define the action of a section $X\otimes s_1 \in \Gamma(M,TM \otimes L^n)$ on the space $\Omega^1(M, L^{\otimes m})$ by
$$(X \otimes s_1) \cdot (\omega \otimes s_2) = (\iota_X \omega) (s_1 \otimes s_2) \in \Gamma\big(M,L^{\otimes (n+m)}\big),$$
where $\omega \otimes s_2 \in \Omega^1(M, L^{\otimes m})$.

The bracket on the infinite dimensional bundle $ \bigoplus_n (TM \otimes L^{\otimes n})$ over $M$ is then:
\begin{align}
\label{brac}
[ \cdot , \cdot ] : \Gamma(TM \otimes L^{\otimes n}) \times \Gamma(TM \otimes L^{\otimes m}) &\to  \Gamma\big(TM \otimes L^{\otimes (m+n)}\big)  \\
( X \otimes s_1 , Y \otimes s_2 ) &\mapsto [ [\mathbbm{1} \otimes \nabla^{\otimes p} + d \otimes \mathbbm{1}, X \otimes s_1], Y \otimes s_2 ],\nonumber
\end{align}
where $X \otimes s_1 \in \Gamma(TM \otimes L^{\otimes n})$, $Y \otimes s_2 \in \Gamma(TM \otimes L^{\otimes m})$, $\mathbbm{1}$ denotes the identity map, the bracket on the right hand side consists of weighted commutators and acts on sections of the bundle $TM \otimes L^{\otimes (m+n)} \to M$, and the value of $p$ changes depending on the sections it acts on, i.e. $p=n$ when acting on $s_1$ and $p=m$ when acting on $s_2$. This can be viewed as a type of derived bracket.
\end{itemize}

\begin{lemma}
The bracket defined in equation (\ref{brac}) satisfies the conditions:
\begin{itemize}
\item
$[a, b] + [b, a]=0,$
\item
$[a, h b] = a( h) b +h [a, b]$,
\end{itemize}
where $a:=X \otimes s \in \Gamma(TM \otimes L^{\otimes n})$, $b \in \Gamma(TM \otimes L^{\otimes m})$, $c \in \Gamma(TM \otimes L^{\otimes p})$, $h \in \Gamma( L^{\otimes p})$, and $n, m, p \in \mathbb{Z}$.
\end{lemma}

\begin{remark}
The bracket defined in equation (\ref{brac}) does not satisfy the Jacobi identity as a result of the curvature for the connection over $L\to M$ not necessarily being trivial. In the case that the curvature is zero, then the Jacobi identity is indeed satisfied.
\end{remark}

Henceforth, the action and bracket of the bundle $\bigoplus_n (TM \otimes L^{\otimes n})$ over $M$ are taken to be inherent structures of the data $(M, L, \nabla)$.\\

Now we can introduce the concept of an exotic Courant algebroid. But first, it should be noted that throughout this paper, we define:
$$\Gamma \Big(M, \bigoplus_{n \in \mathbb{Z}} E_n \Big) := \underset{n \in \mathbb{Z}}{\bigoplus} \Gamma(M, E_n),$$
for any infinite collection of bundles $\{E_n \}_{n \in \mathbb{Z}}$ over $M$, where the direct sum on the right hand side is taken to be the Fr\'{e}chet space completion of the standard direct sum, (this is the same notation as was used for the exotic differential forms reviewed in section 1.3). This is to be taken as the definition for the infinite direct sum throughout. 
\begin{definition}\label{exotic courant}
Let $M$ be a manifold, $E$ a (complex) Courant algebroid over $M$, $L$ a (complex) line bundle over $M$ with connection $\nabla$, and $\mathcal{L}:=\underset{n \in \mathbb{Z}}{\bigoplus} L^{\otimes n}$.

A (complex) \textit{exotic} Courant algebroid over $M$ is given by an infinite-dimensional vector bundle
$$\mathcal{E} := \underset{n \in \mathbb{Z}}{\bigoplus} (E \otimes L^{\otimes n}) \to M,$$
along with a non-degenerate bilinear map 
$\langle \, \, , \, \rangle : \Gamma(M,\mathcal{E}) \times \Gamma(M,\mathcal{E}) \to \Gamma(M,\mathcal{L}),$ a bilinear bracket 
$[ \, \, , \, ]: \Gamma(M,\mathcal{E}) \times \Gamma(M,\mathcal{E}) \to \Gamma(M,\mathcal{E}),$ a bundle map (called the exotic anchor map) $\rho : \mathcal{E} \to \underset{n \in \mathbb{Z}}{\bigoplus} (TM \otimes L^{\otimes n })$, and an induced differential operator 
$D: \Gamma(M,\mathcal{L}) \to \Gamma(M,\mathcal{E})$
 defined by the relation 
 $$\langle Dh, a\rangle = \rho (a) h,$$
 such that the following properties are satisfied:
\begin{enumerate}
\item 
$[a, [b, c]] = [[a, b ], c] + [b, [a, c]]$,
\item 
$\rho([a, b]) = [\rho(a) , \rho(b)]_*,$
\item
$[a, h b] = \rho(a)(h) b +h[a, b]$,
\item
$[a, b] + [b, a] =D \langle a , b \rangle$,
\item
$\rho(a) \langle b , c \rangle = \langle [a, b] , c \rangle + \langle b , [a, c] \rangle$.
\end{enumerate}
where $a, b, c \in \Gamma(M,\mathcal{E})$, $h \in \Gamma(M,\mathcal{L})$, and $[ \, , \, ]_*$ denotes the bracket given in equation (\ref{brac}).
\end{definition}
\begin{remark}
The operator $D$ can be extended to an operator acting on the sections of the infinite vector bundle by the following:
$$(d \otimes \mathbbm{1} + \mathbbm{1} \otimes D): \Gamma(M,\mathcal{E}) \to \Gamma(M,\mathcal{E}).$$

We will often write this extension operator simply as $D$, and from here on out this is what we will mean by $D$ acting on $\Gamma(M,\mathcal{E})$.
\end{remark}

From here on out, we will call sections of the bundle $E \otimes L^{\otimes n} \to M$, sections of the $n$-th weight space of the bundle associated to the exotic Courant algebroid given in definition \ref{exotic courant}. This is in alignment with the terminology used throughout \cite{HM18}.

\subsection{Exotic Courant algebroid examples}
We are now going to introduce two complex, exotic Courant algebroids defined over a principal circle bundle.
We motivate the first example from the definition of the standard Courant algebroid over a principal circle bundle by showing that it is equivalent to an exotic Courant algebroid through the consideration of Fourier expansions about the circle dimension.\\

Before we can introduce the two examples, we will first fix the background data. Let $\bar{\pi} : TZ \oplus T^*Z \to Z$ denote the standard Courant algebroid for $Z$, where $\pi : Z \to M$ denotes a $S^1$-principal bundle, and let $\{U_{\alpha}\}_{\alpha \in \mathcal{A}}$ denote a good cover of $M$ such that $\pi^{-1}(U_{\alpha}) \cong U_{\alpha} \times S^1$. On the bundle $Z$, take $A \in \Omega^1(Z; \mathbb{R})$ to denote a non-normalized connection form, which we can express locally over $\pi^{-1}(U_{\alpha})$ by:
$$A|_{\alpha} = 2 \pi i d \theta + A^*_{\alpha},$$
where $\theta$ denotes the local circle coordinate and $A^*_{\alpha}$ is a local basic form with respect to the bundle $\pi$.

Finally, take $p :L \to M$ to denote the complex line bundle associated to $Z$, and let $s_{\alpha} \in \Gamma(U_{\alpha}, L|_{U_{\alpha}})$ denote the local, nowhere-zero sections of $L$ corresponding to the constant map $U_{\alpha} \to \{1\} \subset S^1$.

Similarly, for the dual data we take all the same notations as above, but now with a hat, i.e. $\hat{ Z}$, to distinguish it from the original.

\begin{example}
\label{ex1}
Consider the infinite dimensional bundle over $M$ given by
\begin{align*}
\mathcal{E}=\bigoplus_{n \in \mathbb{Z}} E_n &:= \bigoplus_{n \in \mathbb{Z}}  \left( (TM \oplus \mathbbm{1}_{\mathbb{C}} \oplus T^*M \oplus \mathbbm{1}_{\mathbb{C}}) \otimes L^{\otimes n} \right).
\end{align*}
Before we define the exotic structure on this bundle, we will utilize a bijection between the sections of this bundle over $M$ and the invariant sections of a bundle defined over $Z$. This bijection is given by
\begin{align}
\label{psi}
\psi : \Gamma \big(M, (TM \oplus \mathbbm{1}_{\mathbb{C}} \oplus T^*M \oplus \mathbbm{1}_{\mathbb{C}}) \otimes L^{\otimes n} \big) &\to \Gamma \big(Z, (TZ \oplus T^*Z) \otimes \pi^*(L^{\otimes n}) \big)^{ S^1} \, \, \, 
\end{align} 
which can be expressed locally over each weight space by
\begin{align*}
\psi \big( (X_{n, \alpha}, f_{n, \alpha}, \xi_{n, \alpha}, g_{n, \alpha})\otimes s_{\alpha}^{\otimes n} \big) |_{\pi^{-1}(U_{\alpha})} = (X_{n, \alpha} + f_{n, \alpha} v + \xi_{n, \alpha}+ g_{n, \alpha}A)\otimes \pi^*(s_{\alpha}^{\otimes n}).
\end{align*}

Utilizing this map, we will (in poor but convenient form) often refer to the sections over $M$ whilst instead writing down the corresponding invariant section over $Z$. This will be done when the structures are simplest to express as sections over $Z$, but it should be understood that what is really meant is the sections over $M$.

Now we define the following structures:
\begin{itemize} 
\item \textit{Differential operator:} Letting $\nabla^L$ denote the induced connection on $L$ (induced from the connection $A$ on $Z$), take $D = \bigoplus_{n \in \mathbb{Z}} D_n$ where:
$$D_n = \nabla^{L^{ \otimes n}} - nA : \Gamma(M,L^{\otimes n})\to \Gamma(M,E_n). $$
This operator\footnote{Again, observe that we incorrectly expressed the map $D_n$ as a map to sections over $Z$, while what we implicitly mean is the map $\psi^{-1} \circ D_n$. This map extends to a map from $\Gamma(M,\mathcal{E}) \to \Gamma(M,\mathcal{E})$ via considering the map $\psi^{-1} \circ (d \otimes \mathbbm{1} + \mathbbm{1} \otimes D) \otimes \psi$.}
(along with the twisted operator $D+H$ where $H$ is an invariant closed 3-form), squares to zero. This is clear when the operator is expanded in local coordinates, say over $\pi^{-1}(U_{\alpha})$, as then: $(\nabla^{L^{ \otimes n}} - nA)^2 = (d - 2 \pi i n d\theta_{\alpha})^2 =0$.

\item{\textit{Exotic bracket:}} The bracket on sections of the weighted spaces $E_p$ for all $p \in \mathbb{Z}$ is the $H$-twisted, derived bracket of the above differential:
\begin{align*}
[ \, \cdot \, , \cdot \, ]_H : \Gamma(M,E_n) \times \Gamma(M,E_m) &\to \Gamma(M,E_{n+m})\\
(a , b) \, &\mapsto [a, b]_H := [ [d \otimes \mathbbm{1} + \mathbbm{1} \otimes D+H \otimes \mathbbm{1}, a ], b ],
\end{align*}
where we have once again implicitly utilized $\psi$ in our definition, and $n, m \in \mathbb{Z}$.

\item{\textit{Bilinear product:}} The $C^{\infty}(M)$-bilinear product on $E$ which, when restricted to the weight spaces $E_p$, is given by:
\begin{align*}
\langle \, \, \, , \, \, \rangle \, : \, \, \Gamma(M,E_n) \times \Gamma(M,E_m) & \to \quad  \Gamma\big(M,L^{\otimes (m+n)}\big)\\
\big( a \otimes s_1^{\otimes n} , b \otimes s_2^{\otimes m} \big) & \mapsto \langle a, b \rangle_{C.A} \big( s_1^{\otimes n} \otimes s_2^{\otimes m} \big),
\end{align*}
where $\langle \, \, , \, \rangle_{C.A}$ denotes the inner product on the Courant algebroid $E=TM \oplus \mathbbm{1}_{\mathbb{C}} \oplus T^*M \oplus \mathbbm{1}_{\mathbb{C}}$ over $M$ defined in Example \ref{ICA1}.

\item{\textit{Anchor map:}} The bundle map $\rho = \oplus_n \rho_n$, where
\begin{align*}
\rho_n : E_n  \to TM \otimes L^{\otimes n}
\end{align*}
is the bundle map defined by the relation:
$$\langle a_n , D h \rangle = \rho_n(a_n) h,$$
where $h \in \bigoplus_m \Gamma(M,L^{\otimes m})$ and $a_n \in \Gamma(M,E_n)$, for all $n, m \in \mathbb{Z}$.
\end{itemize}
\end{example}

\subsubsection{Exotic bracket via geometric means:}
For the exotic Courant algebroid defined in Example \ref{ex1}, we can also derive the $H$-twisted standard exotic bracket via geometric means by considering the exotic twisted Lie derivative along $\mathcal{U} = (X + \alpha)\otimes s \in \Gamma \left( (TM \oplus T^*M) \otimes L^{\otimes n} \right)$ given by $\mathcal{L}_{\mathcal{U}} = \bigoplus_{n \in \mathbb{Z}} \mathcal{L}_{\mathcal{U}}^{L^{\otimes n}}$, where
\begin{align*}
\mathcal{L}_{\mathcal{U}}^{L^{\otimes n}} = \left( \mathcal{L}_X^{L^{\otimes n}} -\mu_X^{L^{\otimes n}}+ d \alpha - \iota_v \alpha + \iota_{X} H \right) \otimes s + \alpha \wedge \nabla^{\otimes n} s +\nabla^{\otimes n} \circ \iota_X,
\end{align*}
such that 
$$ \mathcal{L}_{\mathcal{U}}^{L^{\otimes n}} : \Omega(M, L^{\otimes m}) \to \Omega(M, L^{\otimes (n+m)}),$$
$\mathcal{L}_X^{L^{\otimes n}}$ is the Lie derivative along the direction $X$ and $\mu_X^{L^{\otimes n}}$ is the moment of the induced connection $\nabla^{\otimes n} : \Omega^*(M, L^{\otimes n}) \to \Omega^{* + 1}(M, L^{\otimes n})$ along the direction $X$.

This map defines an extension of the exotic twisted Lie derivative defined in \cite{HM18}.

\begin{theorem}
Let $\mathcal{U}, \mathcal{V} \in \Gamma\left( M, \bigoplus_{n \in \mathbb{Z}}((TM \oplus T^*M) \otimes L^{\otimes n}) \right),$ and $\nabla := \bigoplus_{n \in \mathbb{Z}} \nabla^{\otimes n}$. Then on $\Omega^*(M, \bigoplus_{m \in \mathbb{Z}} L^{\otimes m})$, we have:
\begin{align*}
&\{\mathcal{U} , \mathcal{V} \} = 2 \langle \mathcal{U} , \mathcal{V} \rangle,\\
&\{ \nabla - \iota_v + H, \mathcal{U}\} = \mathcal{L}_{\mathcal{U}},\\
&[\nabla - \iota_v + H, \mathcal{L}_{\mathcal{U}}] = 0, \text{on } \Omega^*\Big(M , \, \bigoplus_{m \in \mathbb{Z}} L^{\otimes m}\Big)^{S^1},\\
&[\mathcal{L}_{\mathcal{U}}, \mathcal{V} ] = [\mathcal{U}, \mathcal{V}]_H,\\
& [\mathcal{L}_{\mathcal{U}}, \mathcal{L}_{\mathcal{V}}] = \mathcal{L}_{[\mathcal{U}, \mathcal{V}]_H}, \text{ on } \Omega^*\Big(M,\underset{n \in \mathbb{Z}}{{ \bigoplus}}  L^{\otimes m}\Big)^{S^1}.
\end{align*}
\end{theorem}
The proof of this theorem follows the same lines as Theorem 1.2 in \cite{HM18}.

\subsubsection{Relation between the standard Courant algebroid and the exotic Courant algebroid:}

We now detail how the exotic Courant algebroid of Example \ref{ex1} is in fact equivalent to an exact Courant algebroid, such that the data defined on one of these objects can be transferred to the other, and vice versa.

To do this we will use the fact that every element $a \in \Gamma(Z,TZ \oplus T^*Z)$ can be expressed locally over $\pi^{-1}(U_{\alpha}) \cong U_{\alpha} \times S^1$ as a sum over its weight space $\Gamma_n(TZ\oplus T^*Z):= \{a \in \Gamma_n(Z,TZ\oplus T^*Z)| \mathcal{L}_v a = n a\}$,  using the family Fourier expansion:
$$a |_{U_{\alpha}}= \underset{n \in \mathbb{Z}}{\sum} e^{-2 \pi i n \theta_{\alpha}} (X_{n, \alpha} + f_{n, \alpha} v + \xi_{n, \alpha} + g_{n, \alpha} A) \in \underset{n \in \mathbb{Z}}{\bigoplus} \Gamma_{-n} \big( \pi^{-1}(U_{\alpha}), TZ \oplus T^*Z|_{\pi^{-1}(U_{\alpha})}\big),$$
where $X_{n, \alpha} \in \Gamma(U_{\alpha}, TU_{\alpha})$ is a horizontal vector field, $\xi_{n, \alpha} \in \Gamma(U_{\alpha}, T^*U_{\alpha})$ is a basic form, and $f_{n, \alpha}, g_{n, \alpha} \in C^{\infty}(U_{\alpha})$.\footnote{Recall that $X_{n, \alpha}$ should be written as $h_A^{-1}(X_{n, \alpha})$ where $h_A: M \to Z$ denotes the horizontal lift corresponding to the connection $A$, and $\xi_{n, \alpha}$ should be written as $\pi^*(\xi_{n, \alpha})$. Keeping with the literature however, we will take $X_{n, \alpha}$ to mean both the vector field on $TU_{\alpha}$ as well as its horizontal lift defined on $TZ|_{\pi^{-1}(U_{\alpha})}$, with which being clear from the context. Likewise for the basic forms.} 
\begin{proposition}
\label{pro}
The exotic Courant algebroid over $M$ defined in Example \ref{ex1} is equivalent to the H-twisted standard Courant algebroid over $Z$ detailed in Example \ref{stand}.
\end{proposition}
\begin{proof}
Consider the map given by 
$$\phi = \bigoplus_{n \in \mathbb{Z}} \phi_n : \Gamma(Z,TZ \oplus T^*Z) \to \Gamma\big(Z,(TZ \oplus T^*Z) \otimes \pi^*(L^{\otimes n})\big)^{ S^1}$$
that, when decomposed into its weight spaces, defines the maps
\begin{align}
\label{phi}
\phi_n:  \Gamma_{-n}(Z,TZ \oplus T^*Z) \to \Gamma\big(Z,(TZ \oplus T^*Z) \otimes \pi^*(L^{\otimes n})\big)^{ S^1}.
\end{align}
which can locally be expressed over $\pi^{-1}(U_{\alpha})$ by the invertible map:
\begin{align*}
\phi_n(e^{-2 \pi i n \theta_{\alpha}} (X_{n, \alpha} + f_{n, \alpha} v + \xi_{n, \alpha} + g_{n, \alpha} A))|_{\pi^{-1}(U_{\alpha})} &=(X_{n, \alpha} + f_{n, \alpha} v + \xi_{n, \alpha}+ g_{n, \alpha}A)\otimes \pi^*(s_{\alpha}^{\otimes n}).
\end{align*}

Then, recalling the bijective correspondence between the invariant sections over $Z$ and sections over $M$ as defined in equation (\ref{psi}), given by the map:
\begin{align*}
\psi : \Gamma \big(M, (TM \oplus \mathbbm{1}_{\mathbb{C}} \oplus T^*M \oplus \mathbbm{1}_{\mathbb{C}}) \otimes L^{\otimes n} \big) &\to \Gamma \big(Z, (TZ \oplus T^*Z) \otimes \pi^*(L^{\otimes n}) \big)^{ S^1},
\end{align*}
we get a bijection $\phi ' = \psi^{-1} \circ \phi$ which satisfies the following conditions:
\begin{align} 
\label{eq1}
\begin{split}
[ \phi ' (a), \phi '(b) ]_H &= \phi ' ([a, b]_{H,C.A}),\\
\langle \phi '(a), \phi '(b) \rangle &= \phi ''(\langle a, b \rangle_{C.A}),\\
\rho( \phi '(a)) &= \phi ' \circ \rho(a),
\end{split}
\end{align}
where $a , b \in \Gamma(Z,TZ \oplus T^*Z)$, and the map $\phi '': C^{\infty}(Z) \overset{\cong}{\to} \oplus_n \Gamma(L^{\otimes n})$ is given locally by $\phi ''_{\alpha}(e^{-2 \pi i n \theta_{\alpha}})= s_{\alpha}^{\otimes n}$.

Thus, decomposing the Courant algebroid $TZ \oplus T^*Z$ over $Z$ into its weight spaces we can transfer all its data to that of the exotic Courant algebroid. Given this map is invertible, the result is proven.
\end{proof}

\begin{remark}
We will often refer to the above $H$-twisted, standard Courant algebroid $TZ \oplus T^*Z$ over $Z$ as having an exotic Courant algebroid structure, where by this, we mean the exotic Courant algebroid $\mathcal{E}$ over $M$ of which it is equivalent to.
\end{remark}

Keeping all notation the same, and once again letting $\hat{\pi} : \hat{Z} \to M$ denote the T-dual bundle to $Z$ with invariant T-dual H-flux representative $\hat{H} \in \Omega^3(\hat{Z})$, we come to our second example:
\begin{example}
\label{ex2}
Consider the infinite-dimensional bundle over $M$ given by the space
$$\hat{\mathcal{E}} = \bigoplus_{n \in \mathbb{Z}} \hat{E}_n := \bigoplus_{n \in \mathbb{Z}}  \big((T\hat{Z} \oplus T^*\hat{Z}) \otimes \hat{\pi}^*(L^{\otimes n})\big)^{\hat S^1} = \bigoplus_{n \in \mathbb{Z}}  \big( (TM \oplus \mathbbm{1}_{\mathbb{C}} \oplus T^*M \oplus \mathbbm{1}_{\mathbb{C}}) \otimes L^{\otimes n} \big),$$ 
and consider the following structures defined on it:
\begin{itemize}
\item \textit{Differential operator:} The differential operator given by $\hat{D}:= \oplus_n \hat{D}_n$ where:
$$\hat{D}_n = -(\hat{\nabla}^{L^{\otimes n}} - \iota_{n \hat v})  : \Gamma(M,L^{\otimes n}) \to \Gamma(M, \hat{E}_n).$$
The operator $-(\hat{\nabla}^{\otimes n} - \iota_{n \hat v} + \hat{H})$ squares to zero for the invariant closed 3-form $\hat H$, as detailed in Theorem 2.1 of \cite{HM18}.

\item{\textit{Exotic bracket:}} The $\hat H$-twisted derived bracket of the above differential, defined on the weight spaces by:
\begin{align*}
[ \, \cdot \, , \cdot \, ]_{\hat{H}} : \Gamma(M,\hat{E}_n) \times  \Gamma(M,\hat{E}_m) &\to  \Gamma(M,\hat{E}_{n+m})\\
(a , b) \, &\mapsto [a, b]_{\hat{H}} := [ [\mathbbm{1} \otimes \hat D + d \otimes \mathbbm{1} + \hat H \otimes \mathbbm{1}, a ], b ].
\end{align*}

\item{\textit{Bilinear product:}} The $C^{\infty}(M)$-bilinear product which, when restricted to the weight spaces $\hat{E}_n, \hat{E}_m$, is given by:
\begin{align*}
\langle \, \, \, , \, \, \rangle \, : \,  \, \Gamma(M,\hat{E}_n) \times  \Gamma\big(M,\hat{E}_m) & \to \quad   \Gamma(M,L^{\otimes (m+n)}\big)\\
\Big( a \otimes s_1 , b \otimes s_2 \Big) & \mapsto \langle a, b \rangle_{C.A} \big( s_1 \otimes s_2 \big),
\end{align*}
where $ \langle \cdot , \cdot  \rangle_{C.A}$ defines the inner product on sections of the Courant algebroid $TM \oplus \mathbbm{1}_{\mathbb{C}} \oplus T^*M \oplus \mathbbm{1}_{\mathbb{C}}$ over M defined in Example \ref{ICA1}.

\item{\textit{Anchor map:}} The bundle map $\hat{\rho} = \oplus_n \hat{\rho}_n$, where
$$\hat{\rho}_n : \hat{E}_n  \to TM \otimes L^{\otimes n},$$
is the bundle map defined by the relation
$$\langle a_n , \hat{D}h \rangle = \hat{\rho}_n(a_n) h,$$
where $h \in \bigoplus_m \Gamma(M,L^{\otimes m})$ and $a_n \in  \Gamma(M,\hat{E}_n)$.
\end{itemize}

This bundle, along with the above structures, defines an exotic Courant algebroid over $M$. Observe that although the bundle of this example is the same as that of Example \ref{ex1}, the exotic structures defined on the bundles are distinct, and thus the resulting exotic Courant algebroids are also distinct.
\end{example}

\section{T-duality for exotic Courant algebroids}

In this section, we prove the main T-duality result between exotic Courant algebroids in Theorem \ref{T-duality-exotic}, 
that generalizes the T-duality isomorphism 
in \cite{cavalcanti} mapping from the invariant Courant algebroid $(TZ\oplus T^*Z)^{S^1}$ over $M$, to an isomorphism from the exact Courant algebroid $(TZ\oplus T^*Z)$ over $Z$.
We also define a Clifford action of exotic Courant algebroids on exotic differential forms in Theorem \ref{action} and show that it is compatible with T-duality.
 Subsection \ref{trivial} illustrates this isomorphism in the special case of trivial circle bundles with trivial flux.
 
\subsection{T-duality isomorphism from the standard Courant algebroid.}
\hfill\\
Let $\tau = \bigoplus_n \tau_n$ where $\tau_n : \Omega^{\bar{k}}_{-n}(Z) \to \mathcal{A}_n^{\bar{k+1}}(\hat{Z})^{\hat{S}^1}$ denotes the exotic Hori formula as defined in \cite{HM18}, which extends the standard T-duality isomorphisms on invariant differential forms. 

Recall the exotic Courant algebroid types from Example \ref{ex1} and Example \ref{ex2} defined over $M$, each with underlying infinite-dimensional vector bundle given by: 
\begin{align*}
\mathcal{E}_M= \hat{\mathcal{E}}_M := \bigoplus_{n \in \mathbb{Z}} \mathcal{E}_{M,n} &= \bigoplus_{n \in \mathbb{Z}}  \big( (TM \oplus \mathbbm{1}_{\mathbb{C}} \oplus T^*M \oplus \mathbbm{1}_{\mathbb{C}}) \otimes L^{\otimes n} \big),
\end{align*}
where $\mathcal{E}_M$ and $\hat{\mathcal{E}}_M$ denote the same bundle, however to reduce confusion we denote the bundle associated to the exotic Courant algebroid related to the pair $(Z, H)$ from Example \ref{ex1} by $\mathcal{E}_M$, and the bundle associated to the exotic Courant algebroid related to the T-dual pair $(\hat{Z}, \hat{H})$ from Example \ref{ex2} by $\hat{\mathcal{E}}_M$.

Furthermore, sections of these exotic Courant algebroids from Example \ref{ex1} and Example \ref{ex2} can equivalently be viewed as invariant sections of the following infinite-dimensional bundles defined over $Z$ and $\hat{Z}$ respectively, given by:
\begin{align*}
\mathcal{E}_Z :=\bigoplus_{n \in \mathbb{Z}} \mathcal{E}_{Z, n} &= \bigoplus_{n \in \mathbb{Z}}  \big((TZ \oplus T^*Z) \otimes \pi^*(L^{\otimes n})\big),\\
\hat{\mathcal{E}}_{\hat{Z}} := \bigoplus_{n \in \mathbb{Z}} \hat{\mathcal{E}}_{\hat{Z},n} &= \bigoplus_{n \in \mathbb{Z}}  \big((T\hat{Z} \oplus T^*\hat{Z}) \otimes \hat{\pi}^*(L^{\otimes n})\big)
\end{align*}
That is, there exists a bijection $\psi^{-1} : \Gamma(Z, \mathcal{E}_Z)^{S^1} \to \Gamma(M, \mathcal{E}_M)$ as given in (\ref{psi}), and similarly a bijection $\hat{\psi}^{-1}: \Gamma(\hat Z, \hat{\mathcal{E}}_Z)^{S^1} \to \Gamma(M, \hat{\mathcal{E}}_M)$.

Defining the map $\varphi = \oplus_n \varphi_n$ where:
\begin{align*}
\varphi_n : \Gamma(M, \mathcal{E}_{M,n}) &\to \Gamma(M, \hat{\mathcal{E}}_{M,n})\\(X, f, \xi , g) \otimes s^{\otimes n} &\mapsto -(X, g, \xi , f) \otimes s^{\otimes n},
\end{align*}
and recalling the map $\phi ':\Gamma(Z,TZ \oplus T^*Z) \to  \Gamma(M, \mathcal{E}_M)$ from Proposition \ref{pro}, consider the composition:
\begin{align}
\label{mu}
\mu := \varphi \circ \phi '  : \Gamma(Z, TZ \oplus T^*Z) \to \Gamma(M, \hat{\mathcal{E}}_M),
\end{align}
where this map may be decomposed into the sum $\mu = \bigoplus_n \mu_n$ where locally:
\begin{align*}
\mu_n = \varphi_n \circ \phi'_{n} : \Gamma_{-n}(Z, TZ \oplus T^*Z) &\to \Gamma(M, \hat{\mathcal{E}}_{M,n}),\\
e^{-2 \pi i n \theta_{\alpha}} (X_{n, \alpha} + f_{n, \alpha} v + \xi_{n, \alpha} + g_{n, \alpha} A) &\mapsto (X_{n, \alpha}, f_{n, \alpha} v, \xi_{n, \alpha}, g_{n, \alpha}) \otimes s^{\otimes n}
\end{align*}

\begin{theorem}[T-duality of exotic Courant algebroids]\label{T-duality-exotic}
The map defined in equation \eqref{mu} above:
$$
\mu := \Gamma(Z, TZ \oplus T^*Z) \to \Gamma(M, \hat{\mathcal{E}}_M),
$$
is an isomorphism between the exotic Courant algebroid structures of each bundle.
Furthermore, it is equal to the complexified T-duality isomorphism of Cavalcanti-Gualtieri \cite{cavalcanti} when the domain is restricted to the invariant sections of the exact Courant algebroid.
\end{theorem}
\begin{proof}
It was shown in Proposition \ref{pro} that the map $\phi ':\Gamma(Z,TZ \oplus T^*Z) \to  \Gamma(M, \mathcal{E}_M)$ defines an equivalence between the Courant algebroid structure on $TZ\oplus T^*Z $ and the exotic Courant algebroid structure on $\mathcal{E}_{M}$. What we will now show is that $\mu \circ \phi'^{-1} = \varphi$ defines an exotic Courant algebroid isomorphism from $\mathcal{E}_M$ to $\hat{\mathcal{E}}_M$.

First, it is clear that $\varphi$ is invertible, by simply considering the map:
\begin{align*}
\varphi^{-1}_n \, : \, \, \, \,  \Gamma(M, \hat{\mathcal{E}}_{M, n}) \, \, \, \,  &\to \, \, \, \, \Gamma(M, \mathcal{E}_{M, n})\\
(X_{n, \alpha}, f_{n, \alpha}, \xi_{n, \alpha}, g_{n, \alpha})\otimes s_{\alpha}^{\otimes n} &\mapsto -(X_{n, \alpha}, g_{n, \alpha}, \xi_{n, \alpha}, f_{n, \alpha})\otimes s_{\alpha}^{\otimes n}.
\end{align*}

Secondly, $\varphi$ transfers the exotic Courant algebroid structure bijectively, as can be proven explicitly by simply showing:
\begin{align}
\begin{split}
[ \varphi(a), \varphi(b) ]_{\hat{H}} &= \varphi([a, b]_{H}),\\
\langle \varphi(a), \varphi(b) \rangle &= \phi ''(\langle a, b \rangle),\\
\hat{\rho}(\varphi(a)) &= \varphi \circ \rho(a).
\end{split}
\end{align}
Therefore we see that $\varphi$ does indeed define an isomorphism between exotic Courant algebroid. By now composing $\varphi$ and $\phi'$, we get that the map $\mu$ defines an isomorphism between the $H$-twisted standard Courant algebroid over $Z$ and the exotic Courant algebroid corresponding to the dual data $(\hat{Z}, \hat{H})$ detailed in Example \ref{ex2}.

Furthermore, if we restrict to the case of when $n=0$, then we get that all the sections are trivial, $s_{\alpha}=1$ for all $\alpha$, and thus:
\begin{align*}
\mu_0 : \Gamma(Z,TZ \oplus T^*Z)^{S^1} &\to \, \, \Gamma(Z,T\hat{Z} \oplus T^*\hat{Z})^{\hat{S}^1}\\
X + fv + \xi +gA  \, \, \, &\mapsto -(X + g v + \xi + fA),
\end{align*}
which is the T-duality map of Cavalcanti-Gualtieri from Theorem \cite{cavalcanti} (the negative sign arising due to one considering the Courant algebroid with derived bracket of the differential operator $-(d + \hat{H})$ instead  of $d+\hat{H}$).

\end{proof}

Now we are going to show that there is an action of the exotic Courant algebroid on the exotic differential forms such that the extended T-duality isomorphism of \cite{HM18} acts as a module isomorphism. That is, given any $a \in \Gamma(Z,TZ \oplus T^*Z)$, $\omega \in \Omega^*(Z)$,
$$ \tau(a \cdot \omega) = \mu(a) \cdot \tau(\omega).$$

There is a group action of the invariant sections $\Gamma\big(\hat Z , \hat{\mathcal{E}}_{\hat{Z}}\big)^{S^1}$ (and thus, also $\Gamma(M, \hat{\mathcal{E}}_{M})$) on the exotic differential forms $\mathcal{A}^*(\hat{Z})$, given by
$$\left( (\hat X + \hat \alpha) \otimes s_1 \right) \cdot (\hat \omega \otimes s_2) := \left(\iota_{\hat X} \hat \omega + \hat \alpha \wedge \hat \omega \right) \otimes s_1 \otimes s_2.$$
Therefore, we get that
\begin{align*}
\left((\hat X + \hat \alpha) \otimes s_1 \right)^2 \cdot (\hat \omega \otimes s_2) &:= \left((\hat X + \hat \alpha) \otimes s_1 \right) \cdot \left( (\hat X + \hat \alpha) \otimes s_1 \right) \cdot (\hat \omega \otimes s_2)\\
&= \Big{\langle} (\hat X + \hat \alpha)\otimes s_1 \, , (\hat X + \hat \alpha) \otimes s_1 \Big{\rangle} \, \hat \omega \otimes s_2,
\end{align*}
where the product denotes the exotic product defined in Example \ref{ex2}.

As a result, take the algebra to be generated by the set $\Gamma\big(M, \hat{\mathcal{E}}_M\big)$ satisfying the relation
$$\big((\hat{X} + \hat{\xi}) \otimes s_1 \big)^2 = \big{\langle} ( \hat{X}+\hat{\xi} ) \otimes s_1 , (\hat{X}+ \hat{\xi}) \otimes s_1 \big{\rangle}.$$

\begin{theorem}[Clifford action of exotic Courant algebroids on exotic differential forms]\label{action}
The algebra of sections $\Gamma(M, \hat{\mathcal{E}}_M) \cong \Gamma(\hat Z ,\hat{\mathcal{E}}_{\hat{Z}})^{S^1}$ defines an action on the set of exotic differential forms, $\mathcal{A}^{\bar{k+1}}(\hat{Z})^{\hat{S}^1}$, such that the extended T-duality isomorphism mapping
$$\tau :\Omega^{\bar{k}}(Z) \to \mathcal{A}^{\bar{k+1}}(\hat{Z})^{\hat{S}^1}$$ 
defines a module isomorphism.
\end{theorem}

\begin{proof}
Let $a \in \Gamma_{-n}(Z, TZ \oplus T^*Z)$ and $\omega \in \Omega_{-m}^{\bar{k}}(Z)$, such that locally, over $\pi^{-1}(U_{\alpha})$, 
\begin{align*}
a &=e^{-2 \pi in \theta_{\alpha}}(X_{-n,\alpha} + f_{-n, \alpha}v + \xi_{-n, \alpha}+ g_{-n, \alpha} A)\\
\omega &= e^{-2 \pi i m \theta_{\alpha}}(\omega_{-m, \alpha, 1} + \omega_{-m, \alpha, 0} A).
\end{align*}
Then acting $a$ on $\omega$, we get 
\begin{align*}
a \cdot \omega |_{\pi^{-1}(U_{\alpha})} &=e^{-2 \pi i (m+n) \theta_{\alpha}} \big(\iota_{X_{-n}} \omega_{-m,1} + (-1)^{k-1}f_{-n} \omega_{-m,0} + \xi_{-n} \wedge \omega_{-m,1}\\
& \qquad \qquad \qquad \, \, \, + (\iota_{X_{-n}} \omega_{-m,0}  +  \xi_{-n} \wedge \omega_{-m,0} + (-1)^k g_{-n} \omega_{-m, 1}) \wedge A \big),
\end{align*}
where the $\alpha$ subscript has been removed in order to simplify notation. 

Under the extended T-duality isomorphism, $a \cdot \omega$ gets mapped locally to
\begin{align*}
\tau(a \cdot \omega) &=\int^{S^1,n+m} (a_{-n} \cdot \omega_{-m}) \wedge e^{\hat{A} \wedge A} \Big|_{\pi^{-1}(U_{\alpha})}\\
&= (-1)^{k}\Big(\iota_{X_{-n}} \omega_{-m,0}  +  \xi_{-n} \wedge \omega_{-m,0} + (-1)^k g_{-n} \omega_{-m, 1}+\\
& \qquad \qquad \quad+ \left(\iota_{X_{-n}} \omega_{-m,1} + (-1)^{k-1}f_{-n} \omega_{-m,0} + \xi_{-n} \wedge \omega_{-m,1} \right) \wedge \hat{A} \Big) \otimes \hat{\pi}^*(s_{\alpha}^{\otimes (n+m)}).
\end{align*}
Furthermore, under T-duality, 
\begin{align*}
\tau(\omega) &= (-1)^{k-1}\big(\omega_{-m, \alpha, 0} + \omega_{-m, \alpha, 1} \wedge \hat{A} \big) \otimes \pi^*(s_{\alpha}^{\otimes m}),\\
\mu(a) &= -\left( X_{-n,\alpha} + g_{-n, \alpha}v + \xi_{-n, \alpha}+ f_{-n, \alpha} A \right) \otimes \pi^*(s_{\alpha}^{\otimes n}).
\end{align*}
Therefore, acting $\mu(a)$ on $\tau(\omega)$, we find
$$\mu(a) \cdot \tau(\omega) = \tau(a \cdot \omega).$$

Furthermore, taking the inverse extended T-duality mapping $\sigma = -\tau^{-1}$ defined in \cite{HM18}, we get that $\bar{\mu}^{-1}$ defines the relevant module isomorphism for this inverse.
\end{proof}

\subsection{The case of trivial circle bundles}\label{trivial}
Consider now T-duality for the case of the trivial circle bundles. That is, 
$$Z=M\times S^1, \ \ \hat Z=M\times \hat S^1$$ and $H, \hat H$ are equal to $0$ and the connections are all the trivial ones. 

Let $\theta$ be the coordinate function of the circle $S^1$ and $\hat{\theta}$ be the coordinate function of the T-dual circle $\hat S^1$. The connections on each bundle will thus be $d\theta$ and $d\hat \theta$ respectively (up to constant).

To see how T-duality acts on the weighted sections of the standard Courant algebroid over $Z$, take $(x + \alpha)_{-n} \in \Gamma_{-n}(Z,TZ\oplus T^*Z)$. Then globally, $(x + \alpha)_{-n}$ is of the form
$$e^{-2\pi in\theta}(X_{-n} + f_{-n} v + \xi_{-n}+ g_{-n} d\theta),$$ 
where $X_{-n} \in \Gamma(TM)$ is a vector field on $M$, $\xi_{-n} \in \Gamma(T^*M)$ is a 1-form on $M$, $f_{-n}, g_{-n} 
\in C^{\infty}(M)$ and $v= \frac{\partial}{\partial\theta}$ defines the fundamental vector field of the bundle $Z$. 

Then given the dual fundamental vector field, $\hat v= \frac{\partial}{\partial\hat\theta}$,
the section $(x+\alpha)_{-n}$ will be mapped under T-duality as follows: 
\be \label{tautriv}
\mu_n((x+\alpha)_{-n})\,
 = -e^{2\pi in\hat\theta} (X_{-n} + g_{-n}\hat{v} +\xi_{-n}+f_{-n}d \hat\theta) .
\ee
So we see that
$$ \mu_n(x_{-n}+\alpha_{-n}) \in  \Gamma_{n}(T{\hat Z}\oplus T^*{\hat Z}).\\$$

On the other hand, if $(\hat y + \hat \beta)_{-n}\in \Gamma_{-n}(T{\hat Z}\oplus T^*{\hat Z})$, then $(\hat y +\hat\beta)_{-n}$ is globally equal to $$e^{-2\pi in\hat\theta} (\hat Y_{-n} + \hat{f}_{-n}  \hat{v} +\hat \eta_{-n}+\hat{g}_{-n} d\hat\theta ), $$
where $\hat{Y}_{-n} \in \Gamma(TM)$ is a vector field on $M$, $\hat{\eta}_{-n} \in \Gamma(T^*M)$ is a  1-form on $M$, and $\hat{f}_{-n}, g_{-n} \in C^{\infty}(M)$.

Now considering the map  $\mu^{-1} : \Gamma(T{\hat Z}\oplus T^*{\hat Z}) \to \Gamma(TZ \oplus T^*Z)$, we find that 
$$\mu^{-1}_{-n}((\hat y+\hat\beta)_{-n})=- e^{2\pi in \theta}(\hat Y_{-n} + \hat{g}_{-n} v +\hat \eta_{-n}+\hat{f}_{-n}  d\theta). $$
And so evidently:  $ \hat{\mu}^{-1}_{-n}((\hat y +\hat\beta)_{-n})\in  \Gamma_{n}(TZ\oplus T^*Z). $

\end{document}